\documentclass{www2009-submission}
\usepackage{algorithm, algorithmic}
\usepackage{graphicx,epsfig,subfigure}
\usepackage{soul}

\newtheorem{theorem}{Theorem}
\newtheorem{proposition}{Proposition}
\newtheorem{corollary}{Corollary}

\newdef{definition}{Definition}
\newdef{example}{Example}

\begin{document}

\title{Anonymizing Unstructured Data}

\numberofauthors{2}

\author{
\alignauthor Rajeev Motwani\\
\affaddr{Department of Computer Science}\\
\affaddr{Stanford University, Stanford, CA, USA}\\
\email{rajeev@cs.stanford.edu}
\alignauthor Shubha U. Nabar\\
\affaddr{Department of Computer Science}\\
\affaddr{Stanford University, Stanford, CA, USA}\\
\email{sunabar@cs.stanford.edu}
}

\maketitle

\begin{abstract}
In this paper we consider the problem of anonymizing datasets in which each individual is associated with a \textit{set of items} that constitute private information about the individual. Illustrative datasets include market-basket datasets and search engine query logs. We formalize the notion of \textit{$k$-anonymity for set-valued data} as a variant of the $k$-anonymity model for traditional relational datasets. We define an optimization problem that arises from this definition of anonymity and provide $O(k\log k)$ and $O(1)$-approximation algorithms for the same. We demonstrate applicability of our algorithms to the America Online query log dataset.
\end{abstract}

\section{Introduction}
\vspace{0.1in}
Consider a dataset containing detailed information about the private actions of individuals, e.g., a market-basket data- set or a dataset of search engine query logs. Market-basket datasets contain information about items bought by individuals and search engine query logs contain detailed information about the queries posed by users and the results that were clicked on. There is often a need to publish such data for research purposes. Market-basket data, for instance, could be used for association rule mining and for the design and testing of recommendation systems. Query logs could be used to study patterns of query refinement, develop algorithms for query suggestion and improve the overall quality of search. 

The publication of such data, however, poses a challenge as far as the privacy of individual users is concerned. Even after removing all personal characteristics of individuals such as actual usernames and ip addresses, the publication of such data is still subject to privacy attacks from attackers with partial knowledge of the private actions of individuals. 
Our work in this paper is motivated by two such recent data releases and privacy attacks on them. 

In August of 2006, America Online (AOL) released a large portion of its search engine query logs for research purposes. The dataset contained $20$ million queries posed by $650,000$ AOL users over a $3$ month period. Before releasing the data, AOL ran a simplistic anonymization procedure wherein every username was replaced by a random identifier. Despite this basic protective measure, the New York Times~\cite{nyt} demonstrated how the queries  themselves could essentially reveal the identities of users. For example, user $4417749$ revealed herself to be a resident of Gwinnett County in Lilburn, GA, by querying for businesses and services in the area. She further revealed her last name by querying for relatives. There were only 14 citizens with her last name in Gwinnett County, and the user was quickly revealed to be Thelma Arnold, a 62 year old woman living in Georgia. From this point on, researchers at the New York Times could look at all of the queries posed by Ms. Arnold over the 3 month period. The publication of the query log data thus constituted a very serious privacy breach.

In October of 2006, Netflix announced the \$$1$-million Netflix Prize for improving their movie recommendation system. As a part of the contest Netflix publicly released a dataset containing $100$ million movie ratings created by $500,000$ Netflix subscribers over a period of $6$ years. Once again, a simplistic anonymization procedure of replacing usernames with random identifiers was used prior to the release. Nevertheless, it was shown that $84$\% of the subscribers could be uniquely identified by an attacker who knew $6$ out of $8$ movies that the subscriber had rated outside of the top $500$~\cite{arvind}. 

The commonality between the AOL and Netflix datasets is that each individual's data is essentially a set of items. Further this set of items is both identifying of the individual as well as private information about the individual, and partial knowledge of this set of items is used in the privacy attack. In the case of the Netflix data (representative of market-basket data), for instance, it is the set of movies that a subscriber rated, and in the case of the AOL data, it is the set of queries that a user posed, also called the \textit{user session}.

Motivated by these examples, as well as by the very real need for releasing such datasets for research purposes, we propose a notion of anonymity for set-valued data in this paper. Informally, a dataset is said to be $k$-anonymous if every individual's ``set of items" is identical to those of at least $k-1$ other individuals. So a user in the Netflix dataset would be $k$-anonymous if at least $k-1$ other users rated exactly the same set of movies; a user in the AOL query logs would be $k$-anonymous if at least $k-1$ other users posed exactly the same set of queries.

One simple way to achieve $k$-anonymity for a dataset would be to simply remove every item from every user's set, or to add every item from the universe of items to every single set. Naturally this would radically distort the dataset rendering it useless for analyses. So instead, to provide greater utility than such a simplistic scheme, we seek to make the minimal number of changes possible to the dataset in order to achieve the anonymity requirements. We provide $O(k\log k)$ and $O(1)$-approximation algorithms for this optimization problem. Further we demonstrate how these algorithms can be scaled for application to massive modern day datasets such as the AOL query logs. To summarize our contributions.

\begin{itemize}
\item We define the notion of $k$-anonymity for set-valued data and introduce an optimization problem for minimally achieving $k$-anonymity in Section~\ref{sec:kanondefn}.
\item We provide algorithms with approximation factors of $O(k\log k)$ and $O(1)$ for the optimization problem in Section~\ref{approxalgs}.
\item In Section~\ref{aolexp}, we demonstrate how our algorithms can be scaled for application to massive datasets and experiment on the AOL logs . 
\end{itemize}

Before proceeding further, note that illustrative datasets used as motivating examples above also contain further user information: time stamp information for when a rating was given and the actual rating itself in the Netflix data; time stamp information for when a query was posed and the query result that was clicked on in the AOL data. However for the purposes of this paper, we ignore these other attributes of the dataset and discuss how they could potentially be dealt with in Section~\ref{discussion}. Indeed the privacy attacks mentioned above did not involve knowledge of these other attributes, and therefore the anonymization problem on even just the reduced set of attributes is important to study. 
 
We will next briefly review related work where we distinguish our problem from the traditional $k$-anonymity problem that has been studied for relational datasets.

\section{Related Work}
\vspace{0.1in}
There has been considerable prior work on anonymizing traditional relational datasets such as medical records. The most widely studied anonymity definitions for such datasets are \textit{$k$-anonymity}~\cite{aggarwal,meyerson,park,sweeney,lefevre} and its variants, \textit{$l$-diversity} ~\cite{ashwin} and \textit{$t$-closeness}~\cite{suresh}. In all these definitions, certain public attributes of the dataset are initially determined to be ``quasi-identifiers''. For instance, in a dataset of medical records, attributes such as Date-of-Birth, Gender and Zipcode would qualify as quasi-identifiers since in combination they can be used to uniquely identify 87\% of the U.S. population~\cite{sweeney}. A dataset is then said to be $k$-anonymous if every record in the dataset is identical to at least $k-1$ other records on its quasi-identifying attribute values. The idea is that privacy is achieved if every individual is hidden in a crowd of size at least $k$. Anonymization algorithms achieve the $k$-anonymity requirement by \textit{suppressing} and \textit{generalizing} the quasi-identifying attribute values of records. A trivial way to achieve $k$-anonymity would be to simply suppress every single attribute value in the dataset, but this would completely destroy the {\em utility} of the dataset. Instead, in order to preserve utility, the algorithms attempt to achieve the anonymity requirement with a minimum number of suppressions and generalizations. 

The kinds of datasets that we consider in this paper differ from traditional relational datasets in two ways. First, each database record in our scenario essentially corresponds to a set of items. The database records could thus be of variable length and high dimensionality. Further, there is no longer a clear distinction between private attributes and quasi identifiers. A user's queries are both private information about the user as well as identifying of the user himself. Similarly, in the case of market-basket data, the set of items bought by an individual are private information about the individual and at the same time can be used to identify the individual. Our definition of anonymity and anonymization algorithms are applicable for such set-valued data.

In \cite{kalnis} the authors study the problem of anonymizing market-basket data. They propose a notion of anonymity similar to $k$-anonymity where a limit is placed on the number of private items of any individual that could be known to an attacker beforehand. The authors provide generalization algorithms to achieve the anonymity requirements. For example, an item `milk' in a user's basket may be generalized to `dairy product' in order to protect it. In contrast, the techniques we propose consider additions and deletions to the dataset instead of generalizations. Further, we demonstrate applicability of our algorithms to search engine query log data as well where there is no obvious underlying hierarchy that can be used to generalize queries. 

Our $O(1)$-approximation algorithm is derived by reducing the anonymization problem to a clustering problem. Clustering techniques for achieving anonymity have also been studied in~\cite{clusteringanonymity}, however here the authors seek to minimize the maximum radius of the clustering, whereas we wish to minimize the sum of the Hamming distances of points to their cluster centers.

In \cite{philip} the authors propose the notion of $(h, k, p)$-coherence for anonymizing transactional data. Here once again there is a division of items into public and private items. The goal of the anonymization is to ensure that for any set of $p$ public items, either no transaction contains this set, or at least $k$ transactions contain it, and no more than $h$ percent of these transactions contain a common private item. The authors consider the minimal number of suppressions required to achieve these anonymity goals, however no theoretical guarantees are given.

Besides the $k$-anonymization based techniques, there has also been considerable work on anonymizing datasets by the addition of noise or perturbation ~\cite{agrawalsrikant, alexandre, dilys}. We do not consider perturbation-based approaches in this paper. 

With regards to search engine query logs, there has been work on identifying privacy attacks both on users~\cite{jasmine} as well as on companies whose websites appear in query results and get clicked on~\cite{ricardo}. We do not consider the latter kind of privacy attack in this paper. \cite{jasmine} considers an anonymization procedure wherein keywords in queries are replaced by secure hashes. The authors show that such a procedure is susceptible to statistical attacks on the hashed keywords, leading to privacy breaches. There has also been work on defending against privacy attacks on users in~\cite{eytan}. This line of work considers heuristics such as the removal of infrequent queries and develops methods to apply such techniques on the fly as new queries are posed. In contrast, we consider a static scenario wherein a search engine would like to publicly release an existing set of query logs.  

\section{Definitions}\label{sec:kanondefn}
\vspace{0.1in}
Let $D = \{S_1, \ldots, S_n\}$ be a dataset containing $n$ records. Each record $S_i$ is a set of items. Formally $S_i$ is a non-empty subset of a universe of items, $U = \{e_1, e_2, \ldots, e_m\}$. We can then define an anonymous dataset as follows.

\begin{definition} ($k$-Anonymity for Set-Valued Data)\label{kanondefn}
We say that $D$ is $k$-anonymous if every record $S_i \in D$ is identical to at least $k-1$ other records. 
\end{definition}

Given this definition, we can now define an optimization problem that asks for the minimum number of transformations to be made to a dataset to obtain an anonymized dataset.

\begin{definition}(The $k$-Anonymization Problem for Set-Val- ued Data)
Given a dataset $D = \{S_1, \ldots, S_n\}$, find the minimum number of items that need to be added to or deleted from the sets $S_1, \ldots, S_n$ to ensure that the resulting dataset $D'$ is $k$-anonymous.
\end{definition}

We illustrate the $k$-anonymization problem with an example. 

\begin{example}
Consider the dataset in Figure~\ref{kanonexamplea}. The dataset in Figure~\ref{kanonexampleb} represents a $2$-anonymous transformation that is obtained by making $2$ additions and $1$ deletion. The items $e_3$ and $e_2$ are added to records $S_2$ and $S_3$ respectively while the item $e_6$ is deleted from record $S_4$. The resulting dataset consists of two $2$-anonymous groups: $\{S_1, S_2, S_3\}$ and $\{S_4, S_5\}$. 

\begin{figure}
\begin{center}
\subfigure[Original Dataset]
{\begin{tabular}{|c|c|}
\hline
ID & Contents \\ \hline
$S_1$ & $\{e_1, e_2, e_3\}$\\
$S_2$ & $\{e_1, e_2\}$\\
$S_3$ & $\{e_1, e_3\}$\\
$S_4$ & $\{e_4, e_5, e_6\}$\\
$S_5$ & $\{e_4, e_5\}$\\\hline
\end{tabular}\label{kanonexamplea}}~~~~~
\subfigure[$2$-Anonymous Transformation]
{\begin{tabular}{|c|c|}
\hline
ID & Contents \\ \hline
$S_1$ & $\{e_1, e_2, e_3\}$\\
$S_2$ & $\{e_1, e_2, e_3\}$\\
$S_3$ & $\{e_1, e_2, e_3\}$\\
$S_4$ & $\{e_4, e_5\}$\\
$S_5$ & $\{e_4, e_5\}$\\\hline
\end{tabular}\label{kanonexampleb}}
\end{center}
\caption{$2$-Anonymization}\label{kanonexample}
\end{figure}
\end{example}

As a more concrete example, in the case of market-basket data, the dataset consists of records, where each record is a \textit{basket of items purchased} by an individual. The $k$-anonymization problem then is to add or delete items to individuals' baskets so that every basket is identical to at least $k-1$ other baskets. 

In the case of search engine query logs, the records correspond to user sessions. Instead of treating each user session as a set of queries, we considered a relaxed problem and treat each user session as a \textit{set of query terms or keywords}. See Section~\ref{aolexp} for the details. The $k$-anonymization problem then becomes one of adding or deleting keywords to or from user sessions to ensure that each user session becomes identical to at least $k-1$ other user sessions. Since no two user sessions are likely to be similar on all the queries, we consider a slightly modified problem in our experiments. Each user session is first separated into ``topic-based'' \textit{threads}, and our goal becomes one of anonymizing these threads instead of the original sessions. The result is an increase in the utility of the released dataset. Again, Section~\ref{aolexp} elaborates on the details. 


More generally, the dataset can be thought of as a bipartite graph, with sets (user sessions/baskets/individuals) represented as nodes on the left hand side and items of the universe (keywords searched for/items purchased/movies rated) as nodes on the right hand side. The $k$-anonymization problem then is to add or delete edges in the bipartite graph so that every node on the left hand side is identical to at least $k-1$ other nodes.  

Depending on the application, it may make sense to restrict the set of permissible operations to only additions or only deletions, however in this paper we consider the most general version of the problem that permits both.

\section{Approximation Algorithms}\label{approxalgs}
\vspace{0.1in}
Given these definitions, we are now ready to devise algorithms for optimally achieving $k$-anonymity. We first draw connections between the $k$-anonymization problem for set-valued data and other optimization problems that have previously been studied in literature, namely, the suppression-based $k$-anonymization problem for relational data and the {\em load-balanced facility location problem}. The reductions to these problems automatically give us the approximation algorithms we desire. In what follows we do not describe the algorithms themselves, rather only the reductions. The algorithms can be found in~\cite{meyerson,aggarwal,park,munagala,karger,zoya}. 

A natural question that arises is whether traditional $k$-anonymity algorithms that involve suppressions and generalizations can be used for the $k$-anonymization problem for set-valued data as defined in Section~\ref{kanondefn}. To this end, we first translate the set-valued dataset to a traditional relational dataset.

\subsubsection*{Transforming $D$ to $R_D$}
A dataset $D = \{S_1, \ldots, S_n\}$ can be transformed to a traditional relational dataset $R_D$ by creating a binary attribute for every item $e_i$ in the universe and a tuple for every set $S_i$. Each tuple will then be a vector in $\{0,1\}^m$. The $1$'s correspond to items in the universe that a set contains and the $0$'s correspond to those that it does not\footnote{Note that at no point do our approximation algorithms ever explicitly construct these bit vectors. Rather they operate directly on the set representations of the tuples, computing intersections of pairs of sets. The algorithms therefore scale with the maximum set size rather than $m$. The bit vector representations have only been used here for ease of exposition.}. For example, the dataset from Figure~\ref{kanonexamplea} translates to the dataset in Figure~\ref{reldataset}.

\begin{figure}
\begin{center}
\begin{tabular}{|c|c|c|c|c|c|c|}
\hline
ID & $e_1$ & $e_2$ & $e_3$ & $e_4$ & $e_5$ & $e_6$ \\ \hline
$S_1$ & $1$ & $1$ & $1$ & $0$ & $0$ & $0$ \\
$S_2$ & $1$ & $1$ & $0$ & $0$ & $0$ & $0$ \\
$S_3$ & $1$ & $0$ & $1$ & $0$ & $0$ & $0$ \\
$S_4$ & $0$ & $0$ & $0$ & $1$ & $1$ & $1$ \\
$S_5$ & $0$ & $0$ & $0$ & $1$ & $1$ & $0$ \\
\hline
\end{tabular}
\end{center}
\caption{Dataset from Figure~\ref{kanonexamplea} as a relational dataset}\label{reldataset}
\end{figure}

The $k$-anonymization problem over $D$ now translates to the following problem over $R_D$:

\begin{definition}($k$-Anonymization via Flips)
Given a dataset $R_D$ over a binary alphabet $\{0,1\}$, flip as few $0$'s to $1$'s and $1$'s to $0$'s in $R_D$ as possible so that every tuple is identical to at least $k-1$ other tuples. 
\end{definition} 

It is trivial to see that there is a one-to-one correspondence between feasible solutions for the $k$-anonymization problem over $D$ and the flip-based $k$-anonymization problem over $R_D$. 

\begin{proposition}\label{firstprop}
Any feasible solution, ${\cal S}_{flip}$, to the flip-based $k$-anonymization problem over $R_D$ can be converted to a feasible solution, ${\cal S}_{\pm}$, of the same cost for the $k$-anonymiza- tion problem over $D$ and vice versa.
\end{proposition}

\noindent\textit{Proof Sketch.} For every $0$ that is flipped to a $1$ in ${\cal S}_{flip}$, simply add the corresponding item to the corresponding set in ${\cal S}_{\pm}$, and for every $1$ that is flipped to a $0$, delete the item from the set.\\

Now the flip-based $k$-anonymization problem can be solved using suppression-based $k$-anonymization techniques for traditional relational datasets studied in~\cite{meyerson,aggarwal,park}. The problem studied here essentially boils down to the following. 

\begin{definition}($k$-Anonymization via Suppressions)
Given a dataset $R_D$ over a binary alphabet $\{0,1\}$, what are the minimum number of $0's$ and $1's$ in $R_D$ that need to be converted to *'s to ensure that every tuple is identical to at least $k-1$ other tuples. 
\end{definition}

Now it is easy to see that the following holds. 

\begin{proposition}
Any feasible solution ${\cal S}_*$ to the suppression-based k-anonymization problem can be converted to a feasible flip-based solution ${\cal S}_{flip}$ using Algorithm~\ref{flip*alg}.
\end{proposition}

\begin{algorithm}[H]
\caption{Converting ${\cal S}_*$ to ${\cal S}_{flip}$}\label{flip*alg}
\begin{algorithmic}[1]
\STATE //input: $R_D$, ${\cal S}_*$
\FOR {every $k$-anonymous group of tuples $G$ in ${\cal S}_*$}
\FOR {every column $C$}
\STATE //$C_G$ = $C$ values for rows in $G$ in $R_D$
\IF {number of $1$'s in $C_G >$ number of $0$'s}
\STATE flip the $0$'s in $C_G$ to $1$'s
\ELSE 
\STATE flip the $1$'s in $C_G$ to $0$'s
\ENDIF
\ENDFOR
\ENDFOR
\end{algorithmic}
\end{algorithm}

The algorithm essentially takes every $k$-anonymous group of tuples in ${\cal S}_*$. Then for any column in the group that is suppressed (*ed out), it replaces the column for that group entirely with $1$'s or entirely with $0$'s depending on which action would involve a fewer number of flips in the original dataset $R_D$.

\begin{example}
Figure~\ref{*flipfig} shows an example of an original dataset, a $2$-anonymous dataset ${\cal S_*}$ obtained via suppressions, and a flip-based $2$-anonymous dataset ${\cal S}_{flip}$ obtained by applying Algorithm~\ref{flip*alg} to ${\cal S}_*$. In both the solutions, the two $2$-anonymous groups are $\{S_1, S_4, S_5\}$ and $\{S_2, S_3, S_6\}$.

\begin{figure}
\begin{center}
\subfigure[Original dataset]
{\begin{tabular}{|c|c|c|c|}
\hline
ID & $e_1$ & $e_2$ & $e_3$ \\ \hline
$S_1$ & 1 & 1 & 0 \\
$S_2$ & 0 & 0 & 1 \\
$S_3$ & 1 & 0 & 1 \\
$S_4$ & 1 & 0 & 0 \\
$S_5$ & 1 & 0 & 0 \\
$S_6$ & 1 & 0 & 1\\ \hline
\end{tabular}}\\
\subfigure[${\cal S}_*$]
{\begin{tabular}{|c|c|c|c|}
\hline
ID & $e_1$ & $e_2$ & $e_3$ \\ \hline
$S_1$ & 1 & * & 0 \\
$S_2$ & * & 0 & 1 \\
$S_3$ & * & 0 & 1 \\
$S_4$ & 1 & * & 0 \\
$S_5$ & 1 & * & 0 \\
$S_6$ & * & 0 & 1\\ \hline
\end{tabular}}~~~~~
\subfigure[${\cal S}_{flip}$]
{\begin{tabular}{|c|c|c|c|}
\hline
ID & $e_1$ & $e_2$ & $e_3$ \\ \hline
$S_1$ & 1 & 0 & 0 \\
$S_2$ & 1 & 0 & 1 \\
$S_3$ & 1 & 0 & 1 \\
$S_4$ & 1 & 0 & 0 \\
$S_5$ & 1 & 0 & 0 \\
$S_6$ & 1 & 0 & 1\\ \hline
\end{tabular}}
\end{center}
\caption{${\cal S}_{flip}$ is obtained from ${\cal S}_*$ via Algorithm~\ref{flip*alg}}\label{*flipfig}
\end{figure}
\end{example}

Now we can show the following about Algorithm~\ref{flip*alg}. 

\begin{theorem}\label{firsttheorem}
For a given dataset $R_D$, let the cost of a feasible solution ${\cal S}_*$ to the suppression-based $k$-anonymization problem be within a factor $\alpha$ of the cost of the optimal solution. Then the cost of ${\cal S}_{flip}$ obtained by applying Algorithm~\ref{flip*alg} to ${\cal S}_*$ is within a factor of $O(k\alpha)$ of the cost of the optimal solution for the flip-based $k$-anonymization problem.
\end{theorem}

\begin{proof}
Let $OPT_*$ and $OPT_{flip}$ be the optimal solutions to the suppression-based and flip-based $k$-anonymization problems over $R_D$ respectively. Then it is easy to see that $\text{Cost}(OPT_*) \leq (2k -1)\text{Cost}(OPT_{flip})$. This is because every $k$-anonymous group of tuples in $OPT_{flip}$ consists of at most $2k - 1$ tuples. Further, this group can be converted to a $k$-anonymous group obtained by suppressions by *ing out any column that contains a flip (essentially the reverse of Algorithm~\ref{flip*alg}). 

It is also easy to see that the cost of any solution $S_{flip}$ obtained by applying Algorithm~\ref{flip*alg} to a solution $S_*$ is less than the cost of $S_*$. This gives us the following set of inequalities and our desired result. 

\begin{eqnarray*}
\text{Cost}(S_{flip}) &\leq& \text{Cost}(S_*)\\
& \leq& \alpha \text{Cost}(OPT_*)\\
& \leq &\alpha(2k-1)\text{Cost}(OPT_{flip})
\end{eqnarray*}
\end{proof}

The best possible suppression-based $k$-anonymization algorithm thus gives us a good flip-based anonymization algorithm through the application of Algorithm~\ref{flip*alg}. Since the suppression-based algorithm from~\cite{park} has an approximation ratio of $O(\log k)$, Theorem~\ref{firsttheorem} together with Proposition~\ref{firstprop} gives us the following result.

\begin{corollary}
There exists an $O(k\log k)$-approximation algorithm to the $k$-anonymization problem for set-valued data. 
\end{corollary}

The suppression algorithm from~\cite{park} essentially considers all possible partitions of the dataset into $k$-anonymous groups and chooses a good one using a set-cover type greedy algorithm.

The translation of $D$ to $R_D$ also enables the insight that the $k$-anonymization problem over set-valued data is essentially a clustering problem.  Each set can be viewed as vector in $\{0,1\}^m$. The optimal solution to the following clustering problem then gives us an optimal solution to the $k$-anonymization problem for set-valued data.

\begin{definition}(The $k$-Group Clustering Problem) Given a set of points in $\{0,1\}^m$, cluster the points into groups of size at least $k$ and assign cluster centers in $\{0,1\}^m$ so that the sum of the Hamming distances of the points to their cluster centers is minimized. 
\end{definition}

The following proposition tells us that there is a one-to-one correspondence between feasible solutions to the $k$-group clustering problem and the $k$-anonymization problem for set-valued data. 

\begin{proposition}\label{clusteringprop} Given a solution, ${\cal S}_{group}$, to the $k$-group clustering problem over a dataset $R_D$, we can obtain a solution ${\cal S}_{\pm}$ of the same cost to the $k$-anonymization problem over $D$ and vice versa. 
\end{proposition}

\noindent\textit{Proof Sketch.} For every cluster in ${\cal S}_{group}$, create a $k$-anonym- ous group of the sets corresponding to the cluster points in ${\cal S}_{\pm}$. $k$-anonymity is achieved by adding or deleting items as necessary so that every set in the group becomes identical to the set corresponding to the cluster center.  The sum of the Hamming distances of points to their cluster centers in ${\cal S}_{group}$ thus corresponds to the total number of additions and deletions of items to obtain the solution ${\cal S}_{\pm}$. \\

Given Proposition~\ref{clusteringprop}, we can now focus on solving the $k$-group clustering problem from here on. In this regard, the following result tells us that it suffices to consider potential cluster centers from amongst the data points themselves. 

\begin{theorem}~\label{discrete}
The cost of the optimal solution to the $k$-group clustering problem when the cluster centers are chosen from amongst the set of data points themselves is at most twice the cost of the optimal solution to the $k$-group clustering problem when the cluster centers are allowed to be arbitrary points in $\{0,1\}^m$. 
\end{theorem}

\begin{proof}
Let $OPT$ be the optimal solution to the $k$-group clustering problem when the cluster centers are allowed to be arbitrary points in $\{0,1\}^m$. Now consider a solution ${\cal S}_{rand}$ that maintains the same cluster groups as $OPT$, but replaces each cluster center with a randomly chosen data point from within the cluster. The expected cost of this solution is given below. 
$$ \text{E}[\text{Cost}({\cal S}_{rand})] = \sum_{G \in {\cal G}}\sum_{C \in {\cal C}} \frac{2N^{C_G}_1N^{C_G}_0}{N^{C_G}_1 + N^{C_G}_0}$$
Here ${\cal G}$ is the set of all clusters in ${\cal S}_{rand}$ (which is the same as the set of clusters in $OPT$). ${\cal C}$ is the columns/dimensions of the dataset $R_D$. $N^{C_G}_1$ and $N^{C_G}_0$ are the number of $1$'s and number of $0$'s respectively that the points in a cluster $G$ have in column $C$. The cost of the optimal solution on the other hand is given by
$$\text{Cost}(OPT) = \sum_{G \in {\cal G}}\sum_{C \in {\cal C}} \text{min}(N^{C_G}_1, N^{C_G}_0).$$
By simple algebraic manipulation, it is easy to see that 
$$\text{E}[\text{Cost}({\cal S}_{rand})] \leq 2\text{Cost}(OPT).$$
Since the expected cost of ${\cal S}_{rand}$ is less than twice the cost of $OPT$, there must exist some clustering solution where the cluster centers are chosen from the data points themselves whose cost is less than twice the cost of $OPT$. This completes the proof of the theorem.
\end{proof}

Theorem~\ref{discrete} considerably simplifies the clustering problem since there is now only a linear number of potential cluster centers that need be considered (as opposed to $2^m$). We can now frame this modified $k$-group clustering problem as an integer program. 

\begin{eqnarray*}
\text{min} & \sum_{i,j} x_{ij}d_{ij} &  \\
\text{s.t} & x_{ij} \leq y_j & \forall~i,j\\
 & \sum_i x_{ij} \geq ky_j & \forall~j\\
 & x_{ij}, y_j \in \{0,1\}  & \forall~i,j\\
\end{eqnarray*}

\vspace{-0.15in}
Here $y_j$ is an indicator variable that indicates whether or not data point $S_j$ is chosen as a cluster center. $x_{ij}$ is an indicator variable that indicates whether or not data point $S_i$ is assigned to cluster center $S_j$ and $d_{ij}$ is the Hamming distance between data points $S_i$ and $S_j$. This integer programming formulation is exactly equivalent to the load-balanced facility location problem studied in~\cite{munagala,karger,zoya}. The cluster centers can be thought of as facilities, and the data points as demand points. The task then is to open facilities and assign demand points to opened facilities so that the sum of the distances to the facilities is minimized and every facility has at least $k$ demand points assigned to it. The algorithms for this problem work by solving a modified instance of a regular facility location problem (without the load balancing constraints), and then grouping together facilities that have fewer than $k$ demand points assigned to them. The result from~\cite{zoya} in conjunction with Theorem~\ref{discrete} and Proposition~\ref{clusteringprop}, gives us the following result. 

\begin{theorem}
There exists an O(1)-approximation algorithm for the $k$-anonymization problem for set-valued data.
\end{theorem}

To reemphasize the earlier footnote, the approximation algorithms for suppression-based anonymization or load-bala- nced facility location never need to explicitly compute and operate on the bit vector representations of the records. They can operate directly on the set representations, computing distances between pairs of sets. Algorithm~\ref{flip*alg} need not operate on the bit-vector representations either. It can simply take every $k$-group of sets and add every majority item in the group to all the sets in the group, while deleting other items.
 
\section{Experiments}\label{aolexp}
\vspace{0.1in}
In this section we experimentally demonstrate applicability of our anonymization algorithms to the AOL query log dataset. Recall (Definition~\ref{kanondefn}) that in this dataset records correspond to user sessions and items correspond to the query terms/keywords. As mentioned earlier, the query log dataset also contains other attributes that we ignore in this paper (see Section~\ref{discussion} for a discussion). Our goal then is to add or delete keywords from user sessions so that every session becomes identical to at least $k-1$ others. 

The anonymization algorithms from Section~\ref{approxalgs} cannot be directly applied to the AOL dataset for several reasons: (1) No two users in the dataset are likely to be similar on all their queries since each user session is fairly large, representing $3$ months of queries. The algorithms  when directly applied to the user sessions would thus result in a large number of additions and deletions. (2) The dataset consists of millions of users. The algorithms from Section~\ref{approxalgs} have a quadratic running time and therefore cannot be practically applied to such real world datasets directly. And (3) Different keywords from different users could often be misspellings of each other or derivations from a common stem. The conditions for considering two user sessions to be ``identical'' thus need to be relaxed. 

We describe below the steps we took to overcome these three problems. 

\subsection{Separating User Sessions into Threads}
\vspace{0.1in}
To deal with the issue of large user sessions, we considered a relaxed problem definition: Each user session was first divided into smaller threads and a different random identifier was assigned to each thread. We then considered the anonymization problem over these threads instead of the original sessions. Each user thread was treated as a set of keywords and our goal was to add or delete keywords from user threads so that every user thread became identical to threads from at least $k-1$ other users. 

One trivial way to divide sessions into threads is to treat every single query from a user as a thread of its own and assign a random identifier to it. However this would render the data nearly useless for many forms of analysis (e.g., studying patterns of query refinement). Instead ``topic-based'' threads were determined on the basis of the similarity of constituent queries. For this purpose we employed two simple measures to determine query similiarity:

\begin{itemize}
\item Edit distance: Two queries were deemed similar if the edit distance between them was less than a threshold.
\item Overlapping result sets: Two queries were deemed similar if the result sets returned for each query by a search engine had a large overlap in the top $50$ results. 
\end{itemize}

Using these similarity measures, each user session was separated into multiple threads: Queries in a user session were considered in the order of their time stamps. A query that was similar to one seen before was assigned the same identifier as the previous query. A query that was very different from any of the previously seen queries was assigned a new identifier. This was followed by another round where consecutive threads that contained similar queries were collapsed and so on. This algorithm for determining threads was run on a random sample of $\sim82K$ users who posed a total of $\sim412K$ queries. The $82K$ user sessions were split into $\sim165K$ threads. Each thread had on average $2.55$ unique keywords. 

There may of course exist more sophisticated techniques for separating sessions into topic-based threads, however this is not the focus in this paper. Note that the shift in goal from anonymizing sessions to anonymizing threads, enhances the utility of the released dataset (anonymizing entire sessions would require far too many additions and deletions), without affecting privacy too much. In fact, as we shall see in Section~\ref{casestudy}, the separation into threads itself helps in anonymization.

\subsection{Pre-clustering User Threads}
\vspace{0.1in}
As mentioned earlier, the algorithms from Section~\ref{approxalgs} have a quadratic running time, and cannot be practically applied to our dataset of user threads. To make them more scaleable, we first performed a preliminary clustering step where we clustered similar user threads together using a simple, fast clustering algorithm, and then applied the $k$-anonymization algorithms from Section~\ref{approxalgs} to the threads within each cluster. If a cluster had fewer than $k$ user threads, we simply deleted these threads altogether.  Running the $k$-anonymization algorithms within these small clusters was much more efficient than running them directly on all the user threads at once. 

To do the preliminary clustering, we used the Jaccard coefficient as a similarity measure for user threads. Recall that each thread $S_i$ is a subset of the universe of keywords $U = \{e_1, \ldots, e_m\}$. Under the Jaccard measure, the similarity of two user threads, $S_i$ and $S_j$ is given by 
$$\text{Sim}(S_i, S_j) = \frac{|S_i \cap S_j|}{|S_i \cup S_j|}$$

A straightforward clustering algorithm would involve a comparison between every pair of user threads and would thus be very ineffcient. Instead, to quickly cluster all the user threads, we used Locality Sensitive Hashing (LSH). The LSH technique was introduced in~\cite{lsh} to efficiently solve the nearest-neighbour search problem. The key idea is to hash each user thread using several different hash functions, ensuring that for each function, the probability of collision is much higher for threads that are similar to each other than for those that are different. The Jaccard coefficient as a similarity measure admits an LSH scheme called Min-Hashing~\cite{cohen,broder}.

The basic idea in the Min-Hashing scheme is to randomly permute the universe of keywords $U$, and for each user thread $S_i$, compute its hash value $\text{MH}(S_i)$ as the index of the first item under the permutation that belongs to $S_i$. It can be shown~\cite{cohen,broder} that for a random permutation the probability that two user threads have the same hash function is exactly equal to their Jaccard coefficient. Thus Min-Hashing is a probabilistic clustering algorithm, where each hash bucket corresponds to a cluster that puts together two user threads with probability proportional to their Jaccard coefficient. The LSH algorithm~\cite{lsh} concatenates $p$ hash-keys for users so that the probability that any two users $S_i$ and $S_j$ agree on their concatenated hash-keys is equal to $\text{Sim}(S_i, S_j)^p$. The concatenation of hash-keys thus creates refined clusters with high precision. Typical values for $p$ that we tried were in the range $2-4$. 

Clearly generating random permutations over the universe of keywords and storing them to compute Min-Hash values is not feasible. So instead, we generated a set of $p$ independent, random seed values, one for each Min-Hash function and mapped each user thread to a hash-value computed using the seed. This hash-value serves as a proxy for the index in the random permutation. The approximate Min-Hash values thus computed have properties similar to the ideal Min-Hash value~\cite{approxminhash}. See~\cite{approxminhash} for more details on this technique.

As a result of running the LSH-based clustering algorithm on our user threads, we otained a total of $\sim 84K$ clusters. Each cluster contained an average of $2$ user threads. The largest cluster contained $\sim 2800$ threads and corresponded to the queries that searched for `Google'!

Again, there may exist more sophisticated techniques for clustering similar user threads together, however this is not the focus of this paper, which is meant to be more of a proof of concept.

\subsection{$k$-Anonymity within Clusters}
\vspace{0.1in}
Now within each cluster generated using the LSH scheme above, we ran the $k$-anonymization algorithm from Section~\ref{approxalgs} (i.e., the suppression algorithm from~\cite{park} followed by the application of Algorithm~\ref{flip*alg}).

Before proceeding further, we need to clarify the criterion that was used for deeming two user threads to be identical. As mentioned earlier, different user threads might contain keywords that are actually just misspellings of each other or derivations from a common stem. To deal with this issue, we once again resorted to LSH. We treated each user thread as a set of Locality Sensitive Hashes~\cite{cohen,broder} of its constituent keywords, i.e., a user thread $S_i = \{e_1, \ldots, e_\ell\}$ now became  $S_i = \{\text{LSH}(e_1), \ldots, \text{LSH}(e_\ell)\}$ where $\text{LSH}(e_j)$ is a concatenation of Min-Hashes of the keyword $e_j$\footnote{Each keyword can be treated as a multiset of characters}. Two user threads were considered identical if they had the same set of hashes. 

Now if a $k$-anonymous solution for a particular cluster deemed that a certain LSH value must be deleted from a particular user thread, we simply deleted all the keywords from the user thread that generated that LSH value. If the solution asked for a LSH value to be added to a user thread, we added to the thread one of the keywords \textit{from its cluster} that generated the LSH value. Threads in clusters of size less than $k$ were entirely deleted. 

Figure~\ref{differingk} shows the total number of additions and deletions of keywords that were made for different values of $k$. As would be expected, as $k$ increases, the total number of additions and deletions that need to be made to achieve $k$-anonymity increases. The number of additions is a small fraction of the total cost, and surprisingly goes down as $k$ increases. 

\begin{figure}
\begin{center}
\includegraphics[width = 2.3in,angle = 270]{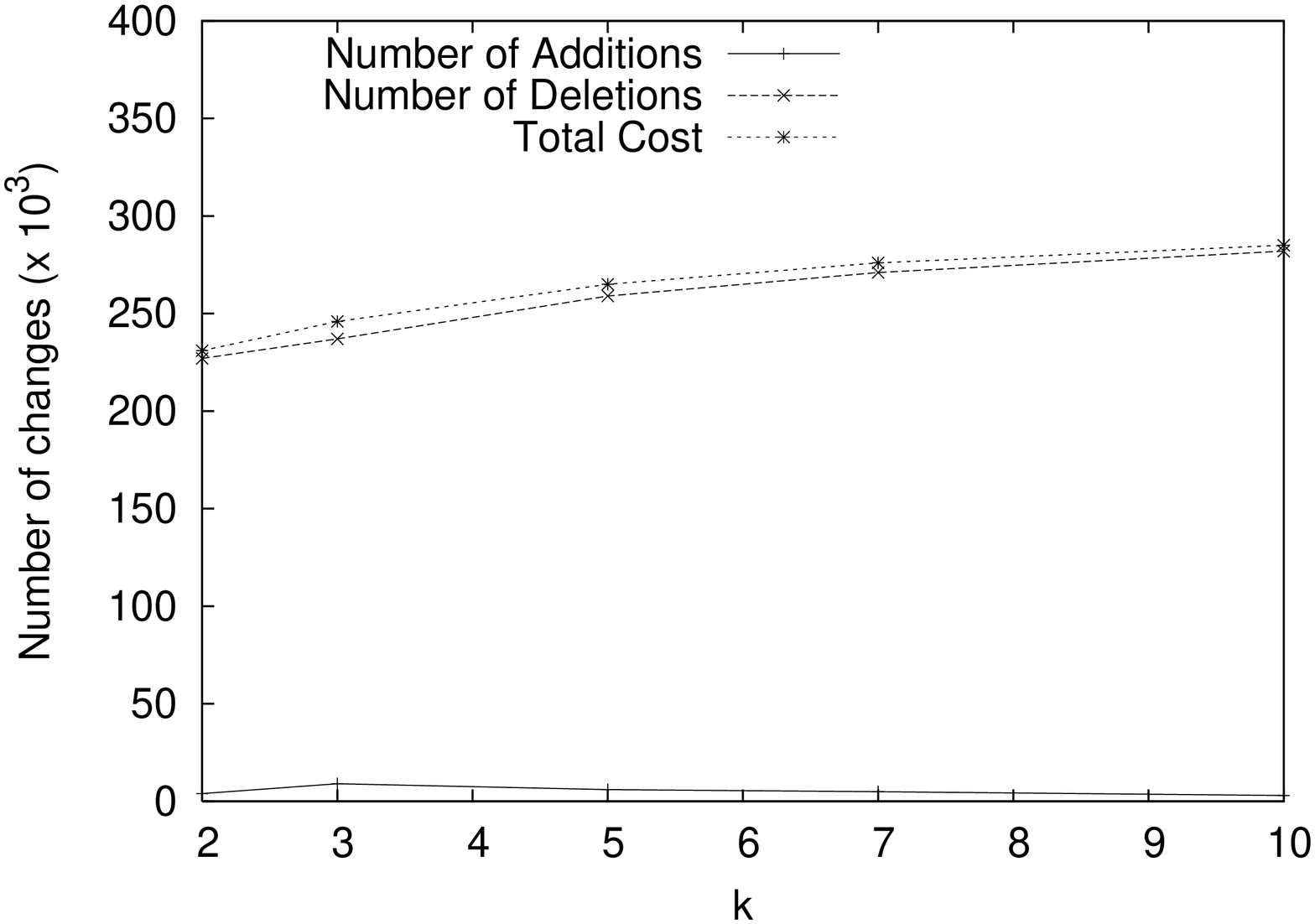}
\caption{Cost of achieving $k$-anonymity}\label{differingk}
\end{center}
\end{figure}

\subsection{Case Study}\label{casestudy}
\vspace{0.1in}
As anecdotal evidence of the effectiveness of our algorithms in anonymizing query logs, we looked at the query logs of user 4417749 who had been previously been identified as Ms. Thelma Arnold from Lilburn, Georgia. 

\begin{figure}
\begin{center}
\subfigure[User 4417749's Session]
{\begin{tabular}{l} 4417749 pine straw lilburn delivery\\
4417749 pine straw delivery in gwinnett county\\
4417749 pine straw in lilburn ga.\\
4417749 atlant humane society\\
4417749 atlanta humane society\\
4417749 dekalb animal shelter\\
4417749 dekalb humane society\\
4417749 gwinnett animal shelter\\
4417749 doraville animal shelter\\
4417749 humane society\\
4417749 gwinnett humane society\\
4417749 seffects of nicotine\\
4417749 effects of nicotine\\
4417749 nicotine effects on the body\\
4417749 jarrett arnold\\
4417749 jarrett t. arnold\\
4417749 jarrett t. arnold eugene oregon\\
4417749 eugene oregon jaylene arnold\\
4417749 jaylene and jarrett arnold eugene or.\\
$\vdots$\\\\
\end{tabular} \label{44a}}
\subfigure[User 4417749's anonymized threads]
{\begin{tabular}{l}\\
1 \st{4417749} pine straw \st{lilburn} \st{delivery} \bf{mulch}\\
1 \st{4417749} pine straw \st{delivery} in \st{gwinnett} \st{county}\\
1 \st{4417749} pine straw in \st{lilburn} \st{ga.}\\
---------------------------------------------\\
2 \st{4417749} atlant humane society \bf{county}\\
2 \st{4417749} atlanta humane society\\
2 \st{4417749} dekalb animal shelter\\
2 \st{4417749} dekalb humane society\\
2 \st{4417749} \st{gwinnett} animal shelter\\
2 \st{4417749} \st{doraville} animal shelter\\
2 \st{4417749} humane society\\
2 \st{4417749} \st{gwinnett} humane society\\
---------------------------------------------\\
3 \st{4417749} seffects of nicotine\\
3 \st{4417749} effects of nicotine\\
3 \st{4417749} nicotine effects on the \st{body}\\
---------------------------------------------\\
4 \st{4417749} \st{jarrett arnold}\\
4 \st{4417749} \st{jarrett t. arnold}\\
4 \st{4417749} \st{jarrett t. arnold eugene oregon}\\
4 \st{4417749} \st{eugene oregon jaylene arnold}\\
4 \st{4417749} \st{jaylene and jarrett arnold eugene or.}\\
$\vdots$\\\\
\end{tabular}\label{44b}}
\end{center}
\caption{User 4417749's Query Logs}\label{4417749}
\end{figure}

Figure~\ref{44a} shows a sample of user 4417749's query logs. Misspellings have been maintained, however repeated queries have been removed. As can be seen, the user searched for some fairly generic queries such as the ``effects of nicotine on the body''. However she also posed several identifying queries. For instance, she queried for humane societies and animal shelters in Gwinnett county, Georgia, revealing herself to be an animal lover in Gwinnett county. Further, she queried for pine straw delivery in Lilburn, Gwinnett, thereby revealing herself to be a resident of Lilburn, Gwinnett. Finally, her queries for relatives in Oregon revealed that her last name was ``Arnold''. 

Figure~\ref{44b} shows the result of running our $k$-anonymizat- ion algorithm for $k=3$. Notice first that the division of Ms. Arnold's session into threads itself goes some way in anonymization by de-correlating her various query topics. The session sample was divided into a thread for pine straw delivery, a thread for animal shelters and humane societies, a thread for the effects of nicotine and a thread for the queries about relatives in Oregon. Each thread was assigned a separate identifier. 

The threads were treated as sets of unique keywords (not depicted in the figure) and were then clustered with the threads of other users using LSH. The anonymization algorithms were run within the resulting clusters. If a particular keyword was to be deleted from a particular thread, we deleted every occurence of that keyword from the original queries of the thread. If a keyword was to be added to a thread, we added it to one of the original queries of the thread. The result was that some threads such as the nicotine thread were left relatively untouched. In the thread for pine straw delivery, the keywords `lilburn', `delivery', `gwinnett', `county' and `ga.' were deleted, and the keyword `mulch' was added instead. This is because other users in the thread's cluster, querying for `pine straw', queried for it in conjunction with the keyword `mulch'. Similarly, in the thread for animal shelters and humane societies, the keywords `gwinnett' and `doraville' were removed, while the keyword `county' was added since many users searched for animal shelters in `dekalb county'. Finally, the thread for the relatives in Oregon was deleted altogether because not a sufficient number of threads from other users got clustered with it. Many users queried for `arnold schwarzenegger', however none of their threads fell in the same cluster!

This example shows that our algorithm does the intuitively right thing. Identifying keywords are removed and keywords that commonly occur in conjunction with other keywords are added to a user's threads. The guarantee is that every user thread will look like the threads of at least $k-1$ other users, and this guarantee is achieved while making a close to minimal number of additions and deletions. 

\subsection{Discussion}~\label{discussion}

While the example of Ms. Thelma Arnold seems to indicate that our anonymization algorithms do the right thing for query logs, our experimental work here is in reality a first step due to the complex nature of the dataset. Several points require further discussion.

\vspace{0.1in}
{\bf Other Attributes:} As mentioned earlier, query logs contain other information about user activity, namely time stamp information for when a query was posed and the query result that was clicked on. Our algorithm focussed on anonymizing just the queries themselves, whereas it is conceivable that these other attributes of the dataset may also be used in launching privacy attacks. One possible anonymization approach is to treat these other attribute values as items of the universe as well and proceed as before. So for example, if a majority of users queried for the Indiana Jones movie on the day that it was released, then this day would be added as part of the time stamp to all user threads on the Indiana Jones movie. The drawback to this approach could be a loss of very fine-grained time stamp information and a better understanding of utility is required before this approach can be recommended. 

\vspace{0.1in}
{\bf Privacy:} In adapting our algorithm to the query logs, we considered a relaxation of the original problem statement: instead of anonymizing entire sessions, we anonymized threads. The privacy implications of this relaxation need to be further examined. At first glance, it seems that the division of the user sessions into threads only helps in our privacy goals by de-correlating a user's query topics. However there is no ``proof'' that a user's threads could not somehow be stitched together to reconstruct his session, which would then no longer be $k$-anonymous. An experimental or theoretical study of the implications of our problem relaxation would be an interesting avenue for future work.

\vspace{0.1in}
{\bf Utility:} Our approach of treating a thread as a set of keywords affects the utility of the released dataset. For example, in Figure~\ref{44b}, the keyword `county' was added to the query for `atlant humane society' since it was to be indiscriminately added to any one of the queries in the thread. In reality it should have been added to the query for `dekalb animal shelter' and that too in the semantically correct position as `dekalb county animal shelter'. Thus by treating threads as sets of keywords, we loose potentially important information about the ordering of keywords within queries. Another point regarding utility, is our criterion for measuring the utility of the released dataset. As in traditional $k$-anonymity work, the criterion we used was to minimize the total number of changes made to the dataset. A better metric for measuring the utility of the released dataset would be to measure the impact of the anonymization on algorithms that actually use the dataset. For example, how well does a search engine's query suggestion algorithm work when run on the released dataset instead of the original. This is a very interesting question, that would need to be ultimately answered for evaluating the utility of any anonymization scheme.

\section{Summary and Future Work}
\vspace{0.1in}
In this paper we introduced the $k$-anonymization problem for set-valued data. Algorithms with approximation factors of $O(k\log k)$ and $O(1)$ for the problem were developed. We applied our anonymization algorithms to the AOL query log dataset. In order to scale the algorithms to deal with the size of the dataset, we proposed a division of the dataset into clusters, followed by the application of anonymization algorithms within the clusters. Besides the problems mentioned in Section~\ref{discussion}, there are several other avenues for future work. For instance, one interesting research direction would be to develop scaleable anonymization algorithms for massive modern day datasets with provable approximation guarantees.  Another important research question is how such algorithms can be applied to anonymize datasets on the fly as new records get added to them. For example, as a search engine receives new queries, how should it anonymize them in an online fashion before storing them.

\section{Acknowledgements}
The authors would like to thank Tomas Feder, Evimaria Terzi and An Zhu for many useful discussions.

\bibliographystyle{abbrv}
\bibliography{setanonbib}

\end{document}